\documentclass[journal,comsoc]{IEEEtran}

\usepackage{fst}
\usepackage[caption=false]{subfig}

\usepackage{mathtools}
\usetikzlibrary{plotmarks}

\usepackage{amsmath}
\usepackage{amsthm}
\usepackage{amssymb}
\usepackage{bm}
\usepackage{xspace}
\usepackage{xcolor}
\usepackage{graphicx}
\usepackage{url}
\usepackage{framed}
\usepackage{float}
\usepackage{rotating}
\usepackage{verbatim}
\usepackage{listings}
\usepackage{lscape}
\usepackage[normalem]{ulem}

\usepackage{pgfplots}
\usepackage{pgf}
\usepackage{tikz}
\usepackage{cite}
\usepackage{csquotes}
\usetikzlibrary{arrows,shapes.misc,chains,scopes}
\usetikzlibrary{matrix}
\usetikzlibrary{calc}
\usetikzlibrary{fit}
\usetikzlibrary{graphs}
\usetikzlibrary{tikzmark}

\usepackage{pgfplotstable}
\usepackage{booktabs}
\usepackage{colortbl}

\newcommand{\executeiffilenewer}[3]{%
\ifnum\pdfstrcmp{\pdffilemoddate{#1}}%
{\pdffilemoddate{#2}}>0%
{\immediate\write18{#3}}\fi%
}

\graphicspath{{figures/}}

\theoremstyle{plain}
\newtheorem{proposition}{Proposition}
\newtheorem{corollary}{Corollary}

\newcounter{examplecount}
\newcounter{remarkcount}

\newenvironment{example}{\refstepcounter{examplecount}\begin{trivlist}\item \textbf{Example \theexamplecount.}}{\end{trivlist}}

\newcommand{\bmm}{\begin{matrix}}
\newcommand{\emm}{\end{matrix}}
\newcommand{\bpm}{\begin{pmatrix}}
\newcommand{\epm}{\end{pmatrix}}

\newcommand{\bsbm}{\left[\begin{smallmatrix}}
\newcommand{\esbm}{\end{smallmatrix}\right]}
\newcommand{\bspm}{\left(\begin{smallmatrix}}
\newcommand{\espm}{\end{smallmatrix}\right)}

\newcommand{\bbm}{\begin{bmatrix}}
\newcommand{\ebm}{\end{bmatrix}}

\newcommand{\mset}[1]{\mathcal{#1}}
\newcommand{\rv}[1]{{#1}}
\newcommand{\pmf}[1]{P_{\rv{#1}}}
\newcommand{\pmfn}[2]{P_{\rv{#1}^{#2}}}

\newcommand{\entrp}[1]{\mathbb{H}\left( #1 \right)}

\newcommand{\expect}[1]{\mathbb{E}\left[ #1 \right]}

\newcommand{\diverg}[2]{\mathbb{D}\left( #1 \Vert #2 \right)}
\newcommand{\selfinfo}[1]{\iota\left( #1 \right)}

\pgfplotsset{compat=newest}

\begin{document}

\title{Divergence-Optimal Fixed-to-Fixed Length Distribution Matching With Shell Mapping}

\author{\IEEEauthorblockN{Patrick Schulte, \IEEEmembership{Student~Member},
Fabian Steiner, \IEEEmembership{Student~Member}}\\
\thanks{P. Schulte and F. Steiner are with the Institute for Communications Engineering, Technical University of Munich (TUM). E-Mails: \{patrick.schulte, fabian.steiner\}@tum.de}
\thanks{This work was supported by the German Federal Ministry of Education and Research in the framework of an Alexander von Humboldt Professorship.}
}

\IEEEtitleabstractindextext{%
\begin{abstract}
Distribution matching (DM) transforms independent and Bernoulli(1/2) distributed bits into a sequence of output symbols with a desired distribution.
A fixed-to-fixed length, invertible DM architecture based on shell mapping is presented.
It is shown that shell mapping for distribution matching (SMDM) is the optimum DM for the informational divergence metric and that finding energy optimal sequences is a special case of divergence minimization. 
Additionally, it is shown how to find the required shell mapping weight function to  approximate arbitrary output distributions.
SMDM is combined with probabilistic amplitude shaping (PAS) to operate close to the Shannon limit.
SMDM exhibits excellent performance for short blocklengths as required by ultra-reliable low-latency (URLLC) applications.
SMDM outperforms constant composition DM (CCDM) by \SI{0.6}{dB} when used with 64-QAM at a spectral efficiency of \SI{3}{bits/channel use} and a 5G low-density parity-check code with a short blocklength of \SI{192}{bits}.
\end{abstract}
}

\maketitle

\IEEEdisplaynontitleabstractindextext
\IEEEpeerreviewmaketitle
\section{Introduction}
Higher-order modulation is a key enabler for high spectral efficiencies and various approaches have been considered in the past to close the shaping gap of discrete constellations with uniformly distributed points (e.g., with \ac{QAM})~\cite{kschischang1993optimal,forney1992trellis,Laroia1994}. Recently, \ac{PAS}~\cite{bocherer2015bandwidth} was proposed that is based on a reverse concatenation architecture~\cite{w._g._bliss_circuitry_1981}, placing the shaping operation before the \ac{FEC} encoding.  Apart from achieving most of the shaping gain, it allows flexible rate adaptation with a single constellation and \ac{FEC} code rate.
To convert uniformly distributed input bits to non-uniformly distributed output symbols, \ac{PAS} requires a \ac{DM}. In \cite{schulte2016constant}, the authors introduced \ac{CCDM} which is asymptotically optimal, in the sense of a vanishing normalized \emph{informational divergence}\cite[p.~7]{kramer2008topics}, for long output blocklengths.

For practical communication systems and new requirements such as \ac{URLLC}, shorter output blocklengths in the range of 10 to 500 symbols are desirable to minimize the processing latency and limit error propagation. Research is therefore now dedicated to find improved \ac{DM} architectures for short blocklengths, e.g., ~\cite{fehenberger2018partition,yoshida_hierarchical_2018}. Good performance for short blocklengths is also needed to operate several \acp{DM} in parallel to further reduce processing latencies.

In this letter, we introduce \ac{SMDM}, a \ac{f2f} length \ac{DM} architecture for short output blocklengths based on shell mapping.
Shell mapping was developed in the early 1990s~\cite{kschischang1992shaping, Laroia1994,khandani1993shaping} and was used in the V.34 modem standard to realize shaping gains with \ac{TCM}.
We show that SMDM minimizes the informational divergence of f2f length \ac{DM} if the self-information of the target distribution is used as the weight function for the shell mapping algorithm.
Further, we show that the dyadic and \ac{MB} distributions~\cite{kschischang1993optimal} lead to integer weight functions, which significantly simplify the implementation of \ac{SM}.
Finally, we explain  how to integrate \ac{SMDM} with \ac{PAS}
to operate close to the Shannon limit at small blocklengths. Numerical simulations with 64-\ac{QAM} and \ac{LDPC} codes from the 5G enhanced mobile broadband (eMBB) standard\cite{richardson_design_2018} show a gain of \SI{0.6}{dB} of \ac{SMDM} over \ac{CCDM} for a \ac{SE} of \SI{3.0}{bits/channel use (bpcu)}.
\section{Preliminaries} \label{sec:prelim}
\subsection{Notation}
We denote random variables with uppercase letters, and their realizations with lowercase letters.
Let $\rv{A}$ be a discrete random variable with \ac{pmf} $\pmf{A}$ defined on the set $\mset{A}$.
If an event $\rv{A} = a$ occurs with positive probability, then its \emph{self-information} is
\begin{equation}
\selfinfo{\pmf{A}(a)} = -\log_2 (\pmf{A}(a)) \text{ bits}.
\end{equation}
The \emph{entropy} of a random variable $\rv{A}$ is the expectation of the self-information of $\rv{A}$, i.e., we have
\begin{align}
\entrp{\pmf{\rv{A}}} &= \expect{\selfinfo{\pmf{A}(\rv{A})}} = \sum\limits_{a\in\supp(\pmf{A})}-\pmf{A}(a)\log_2\left(\pmf{A}(a)\right),
\end{align}
where $\supp(\pmf{A}) \subseteq \mset{A}$ is the support of $\pmf{A}$, i.e., the subset of $a$ in $\mset{A}$ with positive probability.
The informational divergence of two distributions $\pmf{\tilde A}$ and $\pmf{A}$ on $\mset{A}$ is
\begin{equation}
\diverg{\pmf{\tilde{A}}}{\pmf{A}} = \sum_{a \in \supp(\pmf{\tilde{A}})} \pmf{\tilde{A}}(a) \log_2 \frac{\pmf{\tilde{A}}(a)}{\pmf{A}(a)}.
\end{equation}
The mutual information of two random variables $\rv{A}$ and $\rv{B}$ with joint \ac{pmf} $\pmf{\rv{A}\rv{B}}$ is
\begin{equation}
    \mathbb{I}\left( \rv{A}; \rv{B}\right) = \diverg{\pmf{\rv{A}\rv{B}}}{\pmf{\rv{A}}\times\pmf{\rv{B}}},
\end{equation}
with
\begin{equation}
    (\pmf{\rv{A}}\times\pmf{\rv{B}})(ab) = \pmf{\rv{A}}(a) \cdot\pmf{\rv{B}}(a).
\end{equation}
We denote a length $n$ vector of random variables as $\rv{A}^n = [\rv{A}_1\rv{A}_2\ldots\rv{A}_n]$ with realization $a^n = [a_1a_2\ldots a_n]$.
For  random vectors with independent and identically distributed (iid) entries, we write
\begin{equation}
\pmf{A}^n(a^n) = \prod_{i=1}^n \pmf{A}(a_i).
\end{equation}

\subsection{Fixed Length Distribution Matching}
\acp{DM} have applications to capacity achieving communication~\cite{bocherer2015bandwidth} and stealth communication\cite{hou2014effective}.
In both areas, the informational divergence between the output distribution of the \ac{DM} and the target distribution plays a fundamental role.
For energy efficient communication, suppose that $\pmf{A}$ is the capacity-achieving input distribution of a channel with discrete inputs and capacity $C$. Let $\rv{\tilde{Y}}^n$ be the channel output for an input $\rv{\tilde{A}}^n$. Then we have~\cite[Eq.~(23)]{bocherer2011matching}
\begin{equation}
C - \frac{\diverg{\pmfn{\tilde{A}}{n}}{\pmf{A}^n}}{n}\leq \frac{\mathbb{I}(\rv{\tilde{A}}^n;\rv{\tilde{Y}}^n)}{n} \leq C\label{eq:motivation_divergence}.
\end{equation}
Hence, a small divergence guarantees a mutual information close to capacity.

A one-to-one f2f DM is an invertible function $f$ that realizes a desired distribution $\pmf{A}$ on the output symbols.
It maps $m$ uniformly distributed bits $\rv{B}^m$ to length $n$ sequences $\tilde{\rv{A}}^n = f(\rv{B}^m) \in \mset{A}^n$, where $\cA$ is the output alphabet.
The output distribution is defined on a block of $n$ symbols and we denote it by $\pmfn{\tilde{A}}{n}$. 
We call the ratio of input to output lengths the matcher rate 
\begin{equation}
R = \frac{m}{n}.
\end{equation}
In this work, we consider one-to-one f2f distribution matchers.
For vanishing informational divergence, we have (see~\cite{bocherer2014informational})
\begin{equation}
 R = \frac{m}{n}\leq \entrp{\pmf{A}} \label{eq:conv:result}.
\end{equation}


We refer to the image of a DM as the codebook $\mset{C}$ and its elements as codewords.
As a uniformly distributed bit sequence of length $m$ indexes the codewords in the codebook,
every code word has probability $1/|\mset{C}|=2^{-m}$.
Consequently, for the informational divergence the explicit mapping from input to output is not important, only the codebook matters.
We define the \emph{letter distribution} $\pmf{\bar{A}}$ of a codebook as
\begin{equation}
\pmf{\bar{A}}(a) = \frac{1}{|\mset{C}|}\sum_{\alpha^n \in \mset{C}}\frac{n_a(\alpha^n)}{n},\label{eq:def_letterdistrib}
\end{equation}
where $n_a(\alpha^n) = \left| \left\lbrace i: \alpha_i = a \right\rbrace \right|$ is the number of times symbol $a$ appears in $\alpha^n$.
The letter distribution corresponds to the probability of drawing a letter $a$ from the whole codebook.

As each codeword $a^n$ of a codebook $\mset{C}$ is chosen with equal probability, we can write the unnormalized informational divergence as 
\begin{align}
\diverg{\pmfn{\rv{\tilde{A}}}{n}}{\pmf{A}^n} &=-\log_2|\mset{C}| + \sum_{a^n \in \mset{C}} \frac{1}{|\mset{C}|} \selfinfo{\pmf{A}^n(a^n)} \label{eq:CB_divergence_info}\\
&= -\log_2|\mset{C}| + n\entrp{\pmf{\bar{A}}}+n\diverg{\pmf{\bar{A}}}{\pmf{A}}. \label{eq:divergence_letterdistribution}
\end{align}
\subsection{Divergence-optimal codebooks}
To operate close to channel capacity, \eqref{eq:motivation_divergence} suggests to minimize 
$\diverg{\pmfn{\rv{\tilde{A}}}{n}}{\pmf{A}^n}$, where $\pmf{\rv{A}}$ is the optimal input distribution. The optimal codebook  with fixed cardinality $M$ is
\begin{align}
	\hat{\mset{C}}_M &= \argmin_{\substack{\mset{C} \subseteq \mset{A}^n \\ \vert\mset{C}\vert = M}} \diverg{\pmfn{\tilde{A}}{n}}{\pmf{A}^{n}}
	=\argmin_{\substack{\mset{C} \subseteq \mset{A}^n \\ \vert\mset{C}\vert = M}} \sum_{a^n \in \mset{C}}  \sum_{i=1}^n \selfinfo{\pmf{A}(a_i)}\label{eq:DM_optimProblem},
\end{align}
where equality in \eqref{eq:DM_optimProblem} holds because we shifted the objective by $\log_2|\mset{C}|$ and scaled it by $|\mset{C}|$.
Problem \eqref{eq:DM_optimProblem} is solved in \cite{amjad2013algorithms} by selecting those $M$ codewords with the least self-information.
In order find the optimal codebook
\begin{align}
	\hat{\mset{C}} &= \argmin_{\mset{C} \subseteq \mset{A}^n} \label{eq:fullProblem} \diverg{\pmfn{\tilde{A}}{n}}{\pmf{A}^{n}}
	= \argmin_M \left(\argmin_{\substack{\mset{C} \subseteq \mset{A}^n \\ \vert\mset{C}\vert = M}} \diverg{P_{\mset{C}}}{\pmf{A}^{n}}\right)
\end{align}
we need to search through the solutions $\hat{\mset{C}}_M$ of \eqref{eq:DM_optimProblem} for different codebook sizes around $M \approx 2^{n\entrp{\pmf{A}}}$~\cite{amjad2013algorithms} which is not difficult.
We show next how to efficiently encode and decode to $\mset{C}_M$.
\section{Shell Mapping}\label{sec:ShellMapping}
\ac{SM} maps unsigned integers to \emph{shell sequences} $a^n \in \cA^n$, i.e.,
\begin{equation}
\hat{f}_{\text{SM}}: \lbrace 0,1,\ldots,|\mathcal{A}|^n-1 \rbrace \to \mathcal{A}^n.
\end{equation}
We restrict the input to the integers $\lbrace 0,1,\ldots,M-1 \rbrace$ and refer to the image of the shell mapper as the codebook
\begin{equation}
\mset{C}_{\text{SM},M} = \hat{f}_{\text{SM}}(\lbrace0,1,\ldots,M-1\rbrace).
\end{equation}

We assign a non-negative weight $W(a)$ to each letter $a$ in the alphabet $\mset{A}$ using $W\colon \cA \to \setN_0$.
\ac{SM} orders the sequences $a^n \in \mset{A}^n$ according to the sequence weight
$\sum_{i=1}^n W(a_i)$.

This ordering is in general not unique because two sequences may have the same weight, e.g., if they are permutations of each other. A \ac{SM} creates one of these ordered lists. Hence $\mset{C}_{\text{SM},M}$ solves the problem
\begin{equation}
\min_{\substack{\mset{C} \subseteq \mset{A}^n \\ \vert\mset{C}\vert = M}}  \sum_{a^n \in \mset{C}} \sum_{i=1}^n W(a_i),
\label{eq:SM_optimProblem}
\end{equation}
i.e., \ac{SM} finds the set $\mathcal{\hat{C}}$ of $M$ sequences $a^n$ of smallest weight $\sum_{i=1}^n W(a_i)$.
There are many ways to implement \ac{SM}, e.g., using the divide and conquer principle \cite{fischer2002precoding} or sequential encoding \cite{Laroia1994}.
In the next section we show how to use \eqref{eq:SM_optimProblem} for solving \eqref{eq:DM_optimProblem}.
\section{Shell Mapping as Distribution Matcher}\label{sec:smdm}
\subsection{SMDM Interface}
\ac{SM} algorithms require as inputs the codebook cardinality $M$, the output length $n$, and the weight function $W$.
We consider a binary input \ac{DM}, so we choose $M=2^m$ where $m$ is the input blocklength in bits. The input bits are interpreted as an unsigned integer in the range of $\lbrace 0,\ldots, 2^m-1 \rbrace$. The \ac{SM} output corresponds to the output of a \ac{DM}
\begin{equation}
f_{\text{SM},m}: \{0,1\}^m \to \mathcal{A}^n.
\end{equation}
\subsection{Divergence Optimal Weight Functions}
\begin{proposition}
A minimum divergence  f2f length \ac{DM} with a target output probability $\pmf{A}$ is a shell mapper with weight function
\begin{equation}
    \hat{W}(a) = \selfinfo{\pmf{A}(a)}
    \label{eq:weightfunctionForProbabilities}
\end{equation}
\end{proposition}
\begin{proof}
The \ac{SM} algorithm solves problem \eqref{eq:SM_optimProblem}. When we use the \emph{self-information} as a weight function, we solve problem \eqref{eq:DM_optimProblem}.
With a search over the input length, we can find the best distribution matcher independent of the codebook size.
\end{proof}

\begin{example}\label{ex:dyadic}
	Dyadic distributions have the form
	\begin{equation}
	\pmf{A}(a) = 2^{-\ell_a}
	\end{equation}
	where $\ell_a$ is a positive integer for all $a$. The weight function~\eqref{eq:weightfunctionForProbabilities} is
	\begin{equation}\label{eq:DyadicWeightFunction}
	W(a) = \ell_a \quad a \in \supp(\pmf{A}).
	\end{equation}
\end{example}

\textit{Remark:} A non-negative, integer weight function is desirable for implementation. Weight functions constructed with \eqref{eq:weightfunctionForProbabilities} do not generally have integer $\hat{W}(a)$. In the following, we show which practically relevant distributions also yield non-negative integer valued weight functions with \eqref{eq:weightfunctionForProbabilities}.
\begin{proposition}\label{prop:weightFunction}
	Consider a finite support discrete distribution that can be expressed as
	\begin{equation}\label{eq:distrib_families}
	\pmf{A}(a) = \frac{e^{-v\Omega(a)}}{\sum_{\xi \in \supp(\pmf{A})} e^{-v\Omega(\xi)}}, \quad\forall a \in \supp(\pmf{A})
	\end{equation}
	with $v$ being positive and $\Omega$ is any function
	\begin{equation}
	\Omega: \supp(\pmf{A}) \to \mathbb{N}_0,
	\end{equation}
	Then $\Omega$ is a non-negative integer weight function.
\end{proposition}
\begin{proof}
	Inserting \eqref{eq:distrib_families} into \eqref{eq:weightfunctionForProbabilities} we obtain
	\begin{equation*}
	\hat{W}(a) =   v\Omega(a)\log_2(e) + \log_2 \sum_{\xi \in \supp( \pmf{A})} e^{-v\Omega(\xi)} .
	\end{equation*}
	Any translation and positive scaling can be applied on the objective function without changing the codebook. We obtain the integer weight function
	\begin{equation}
	W(a) = \Omega(a) \quad a \in \supp(\pmf{A}).
	\end{equation}
\end{proof}

\begin{example}\label{ex:boltzmann}
	The half \ac{MB} distribution is defined as
	\begin{equation}
	\pmf{A}(a) = \frac{e^{-va^2}}{\sum_{\xi \in \supp(\pmf{A})} e^{-v\xi^2}}
	\end{equation}
	with $\supp(\pmf{A}) = \lbrace 1,3,5,\ldots,2\eta-1 \rbrace =\mset{A}$ and positive $v\in\mathds{R}^+$ and $\eta\in\mathds{N}$.
	Comparing with \eqref{eq:distrib_families} we identify the weight function
	\begin{equation}\label{eq:NonsimpleMBWeightFunction}
	W(a) = a^2 \quad a \in \supp(\pmf{A}) 
	\end{equation}
	which corresponds to the energy of a constellation point.
\end{example} 
\begin{corollary}
We obtain sequences of least power by minimizing divergence to \ac{MB} distributions.
\end{corollary}
This result has a special beauty. \ac{MB} distributions are close to optimal for maximizing the single letter mutual information on discrete signal points for the \ac{AWGN} channel~\cite[Table~1]{bocherer2015bandwidth}. If we minimize the informational divergence of our \ac{f2f} length \ac{DM} to a memoryless source with \ac{MB} distribution, we find that sequences of least energy accomplish this goal. In \cite{kschischang1993optimal}, the authors show that minimizing the average energy subject to a entropy constraint induces \ac{MB} distributed symbols.

The weight function \eqref{eq:NonsimpleMBWeightFunction} is independent of the parameter $v$.
Consequently, according to \eqref{eq:DM_optimProblem} a shell mapper with this weight function implements a minimum divergence DM with fixed codebook size $2^m$ for \emph{all} half MB distributions and
any rate adaptation does not require to change the weight function. A similar property can be observed for distribution families defined in \eqref{eq:distrib_families} for a fixed function $\Omega$ and varying $v$.
\subsection{Determining Letter Frequency $\pmf{\bar{A}}$}\label{sec:letterFrequencies}
A soft-input soft-output decoder requires the letter distribution~\eqref{eq:def_letterdistrib} in order to calculate the priors on the constellation symbols~\cite[Sec. VI-B]{bocherer2015bandwidth}. The letter distribution
depends on both the weight function and how to order sequences of equal weight. Fischer suggests in \cite{fischer1999calculation} an algorithm to calculate the letter distribution.
The algorithm uses the \emph{partial histogram}, i.e., the letter distribution for codebooks that consist of \emph{all} sequences up to a certain weight.
We denote the partial histogram for sequences $a^n$ up to weight $w$ (i.e., $W(a^n) \leq w$) by $\pmf{\bar{A}}(\cdot,w)$, where the first parameter is the symbol that we want to evaluate.
We suggest to use $\pmf{\bar{A}}(\cdot,{w}_\text{max})$ and $\pmf{\bar{A}}(\cdot,{w}_\text{max}-1)$ as approximations of the true frequencies, where $w_{\text{max}}$ is the maximum weight of sequences that the respective SMDM can generate, i.e.,
\begin{equation}
    w_{\text{max}} = W(f_{\text{SM},m}([1,1,\ldots,1])).
\end{equation}
Consider that the all 1 sequence is mapped to a sequence of highest weight. For long blocks we may use the target distribution as approximation at the receiver.
\section{CCDM and SMDM Comparison}
\subsection{Divergence}
\begin{figure}[t]
	\centering
	\footnotesize
	\includegraphics{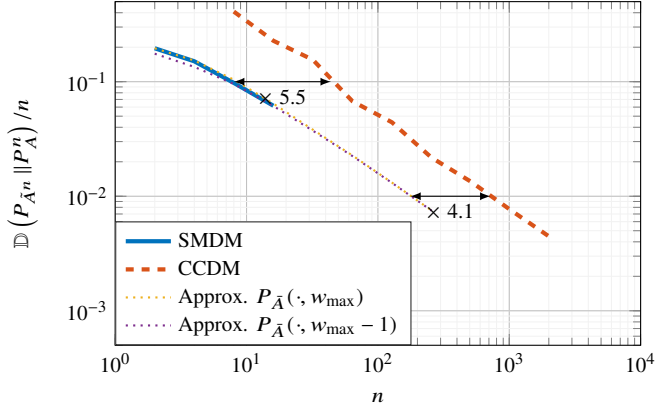}
	\caption{CCDM and SMDM comparison based on the normalized informational divergence and different output blocklengths $n$. The target distribution is 4-ary \ac{MB} and the rate is $R = 1.25$ input bits per output symbol.The calculation of the exact output distribution $P_{\tilde A^n}$ is limited by an input length of $m=  \SI{64}{bits}$ on a standard 64-bit computer architecture.}
	\label{fig:DivergenceComparisonCCDMSMDM}
\end{figure}
To compare CCDM and SMDM we consider the output alphabet $\mathcal{A} = \lbrace 1,3,5,7 \rbrace$ and rate $R = 1.25$. The distribution of the CCDM is an $n$-type approximation \cite[Sec. IV]{bocherer2015bandwidth} of a half MB distribution.
The SMDM has rate $1.25$ and uses the weight function defined in \eqref{eq:NonsimpleMBWeightFunction}.
The results are shown in Fig.~\ref{fig:DivergenceComparisonCCDMSMDM}.
The approximations (dotted lines) use the partial histograms $\pmf{\bar{\rv{A}}}(\cdot, {w}_\text{max}-1)$ and $\pmf{\bar{\rv{A}}}(\cdot, {w}_\text{max})$ as approximations for the letter distribution $\pmf{\bar{\rv{A}}}$ in \eqref{eq:divergence_letterdistribution}.
Using \ac{SMDM} and a target divergence of 0.1 bit we save approximately a factor of $5.5$ in blocklength as compared to \ac{CCDM}, and at a target divergence of $0.01$ we save a factor of 4.1. 
\vspace{-\baselineskip}
\subsection{Rate Adaptation}
Rate adaptation for SMDM is straightforward for distributions of the form~\eqref{eq:distrib_families}. The number of bits that are interpreted as the index of the ordered list can be easily adapted, and therefore the rate can be easily adapted. The granularity of rate adaption is $1/n$, where $n$ is the output length. This granularity is the best possible.%
\vspace{-\baselineskip}
\subsection{Coded Results}
\begin{figure}
	\centering
	\footnotesize
	\includegraphics{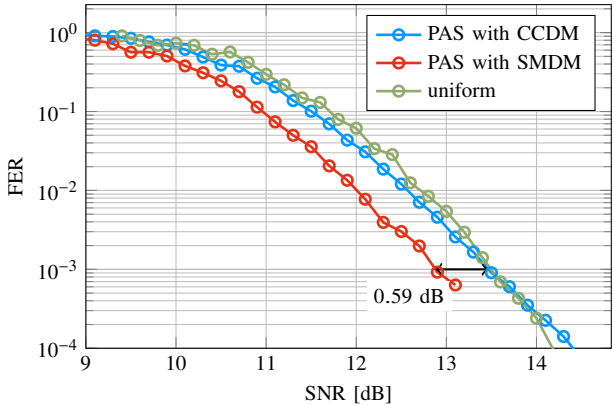}
	\caption{Finite length performance for uniform and shaped signaling using CCDM and SMDM. We target an \ac{SE} of $\SI{3}{bpcu}$ with 5G LDPC codes of blocklength \num{192}.}
	\label{fig:coded_results}
\end{figure}

We compare the performance of \ac{SMDM} and \ac{CCDM} for \ac{PAS} in a coded scenario. We target an \ac{SE} of $\SI{3}{\bpcu}$ with a 64-\ac{QAM} constellation. We employ \ac{LDPC} codes from the recent 5G eMBB standard~\cite{richardson_design_2018} with blocklength \SI{192}{bits}, i.e., 32 complex channel uses. The uniform reference curve uses a rate $R_\tc = 1/2$ code, whereas the shaped scenarios use a rate $R_\tc = 3/4$ code.
Both \acp{DM} approaches have a 4-ary output alphabet to generate the shaped amplitude sequences for the real and imaginary part.
Note that a 64-\ac{QAM} constellation can be constructed as the Cartesian product of two (bipolar) 8-\ac{ASK} constellations, where the latter has four different amplitude values. For 32 complex channel uses with \ac{QAM} symbols, we therefore need 64 amplitudes. The target distribution in both cases is the \ac{MB} family. 
\ac{CCDM} operates with an output blocklength of $n = 64$ output symbols and its performance (blue curve) is similar to the uniform reference (green curve) in Fig.~\ref{fig:coded_results}. The constant composition constraint of \ac{CCDM} thus causes a significant rate loss~\cite[Sec.~V-B]{bocherer2015bandwidth} for small output blocklengths. In contrast, \ac{SMDM} operates with an output blocklength of $n = 32$ (i.e., two \ac{SMDM} are used in parallel) but gains \SI{0.59}{dB} in power efficiency at a frame error rate of $\num{e-3}$.
\section{Conclusion}
We introduced an informational divergence optimal \ac{f2f} length \ac{DM} approach based on \ac{SM}, which shows superior performance compared to state of the art \acp{DM} for short blocklengths. We showed that the self-information of the target output distribution can be used as the weight function for the \ac{SM} algorithm to synthesize arbitrary output distributions. Furthermore, we showed that energy efficient signaling is a special case of divergence minimization. We gave examples for distributions that result in non-negative, integer valued \ac{SM} weight functions favorable for practical implementations.
\section*{Acknowlegement}
The authors would like to thank Gerhard Kramer and Georg B\"ocherer for fruitful discussions.

\ifCLASSOPTIONcaptionsoff
  \newpage
\fi

\end{document}